\documentclass[11pt,times]{article}

\usepackage[utf8]{inputenc}
\usepackage{amssymb,amsmath}
\usepackage{algorithm}
\usepackage{algpseudocode}
\usepackage{amsthm}
\usepackage{enumerate}
\usepackage[font={scriptsize,it}]{caption}
\usepackage{subcaption}
\allowdisplaybreaks
\usepackage[T1]{fontenc}

\usepackage{multirow}
\usepackage{bm}

\usepackage{float}
\newfloat{algorithm}{t}{lop}

\usepackage{tikz}
\usetikzlibrary{decorations.pathmorphing}
\usetikzlibrary{shapes,fit}
\usetikzlibrary{backgrounds}

\algtext*{EndWhile}
\algtext*{EndFor}
\algtext*{EndIf}

\allowdisplaybreaks

\newtheorem{theorem}{Theorem}[section]
\newtheorem{corollary}[theorem]{Corollary}
\newtheorem{definition}[theorem]{Definition}
\newtheorem{lemma}[theorem]{Lemma}

\newtheorem*{theorem*}{Theorem}

\newtheorem*{corollary*}{Corollary}
\newtheorem*{lemma*}{Lemma}
\newtheorem*{observation*}{Observation}

\advance \topmargin by -\headheight
\advance \topmargin by -\headsep
\evensidemargin \oddsidemargin
\begin{document}

\title{Distributed Deterministic Exact Minimum Weight Cycle and Multi Source Shortest Paths in Near Linear Rounds in CONGEST model}

\author{Udit Agarwal}

\maketitle

\begin{abstract}
We present new deterministic algorithms for computing distributed weighted minimum weight cycle (MWC) in  
undirected and directed graphs and distributed weighted all nodes shortest cycle (ANSC) in directed graphs.
Our algorithms for these problems run in $\tilde{O}(n)$ rounds in the CONGEST model on graphs with arbitrary 
non-negative edge 
weights, matching the lower bound up to polylogarithmic factors.
Before our work, no near linear rounds deterministic algorithms were known for these problems.
The previous best bound for solving these problems deterministically requires an initial computation of all pairs
shortest paths (APSP) on the given graph, followed by post-processing of $O(n)$ rounds,
and in total takes $\tilde{O}(n^{4/3})$ rounds, using deterministic APSP~\cite{AR20}.

The main component of our new $\tilde{O}(n)$ rounds 
algorithms is a deterministic technique for constructing a sequence of successive 
blocker sets.
These blocker sets are then treated as source nodes to compute $h$-hop shortest paths, which can then be used to 
compute candidate shortest cycles whose hop length lies in a particular range.
The shortest cycles can then be obtained by selecting the cycle with the minimum weight from all these 
candidate cycles.

Additionally using the above blocker set sequence technique, 
we also obtain $\tilde{O}(n)$ rounds deterministic algorithm
for the multi-source shortest paths problem (MSSP) for both
directed and undirected graphs, given that the size of the 
source set is at most $\sqrt{n}$.
This new result for MSSP can be a step towards obtaining a $o(n^{4/3})$ rounds algorithm for
deterministic APSP.
We also believe that our new blocker set sequence
technique may have potential applications for other distributed algorithms.

\end{abstract}

\section{Introduction}
\subsection{Background and Motivation}

Finding the minimum weight cycle (MWC) 
in a graph is one of the fundamental problems in the area of graph
algorithms and has received a huge amount of attention for the past several decades in the context of 
centralized algorithms (e.g., ~\cite{C15, D20, K22, L72, L18, L09, orlin2017nm, R13, RW11, R12}).
The related problem of finding the shortest cycles for each node in a graph,
known as All Nodes Shortest Cycle (ANSC), is relatively less-studied and 
has only gained attention in the past decade~\cite{A18,D22, S19, Y11}.

In this work, we study the computation of MWC in the CONGEST model of 
distributed computing.
In the CONGEST model, the input is a directed (or undirected) graph $G = (V, E)$, and the distributed 
computation occurs at the nodes in this graph.
The output of MWC problem is to compute the weight of the minimum
weight cycle in the graph $G$.
We assume an arbitrary non-negative edge weight on each edge.
In this paper, we consider the computation of the exact shortest cycles  and paths in the graph.

\subsection{Our Results}

Table~\ref{table-results} lists the comparison of our new results with the earlier 
known deterministic results for the shortest cycle problems, including MWC and ANSC, and
the multi-source shortest paths (MSSP) problem.
We now explain our results in more detail.

\begin{table}
\centering
\caption{Table comparing our new results with previous known deterministic results.} \label{table-results}
\def\arraystretch{1.5}
{
\begin{tabular}{| c | c | c |}
\hline
{\bf Problem} & {\bf Previous Bound} & {\bf This Paper} \\
 \hline
Undirected/ & & \\
Directed MWC & $\tilde{O}(n^{4/3})$~\cite{AR20} & $\bm{\tilde{O}(n)}$ \\
\hline
Directed ANSC & $\tilde{O}(n^{4/3})$~\cite{AR20} & $\bm{\tilde{O}(n)}$ \\
\hline  
Undirected/ &  &  \\
Directed MSSP & $\tilde{O}(\min \{n^{4/3}, n\cdot |S|\})$~\cite{AR20, B58} & $\bm{\tilde{O}(n)}$ \\
\hline
\end{tabular}}
\end{table}

\subsubsection{Directed Shortest Cycles}
In a directed graph, the problem of computing the shortest cycle at a node $v$ involves finding 
the cycle with minimum weight among all cycles that contain node $v$.
The ANSC problem deals with finding the shortest cycle for each node $v$
in the graph, whereas MWC deals with finding the global shortest cycle in the 
graph.
The current known approach for computing shortest cycles involves first computing 
APSP and then for each node $v$, 
taking the minimum across all possible values 
$\{w(v,u) + \delta(u,v) \}$ where $u$ is an outgoing neighbor of $v$.
This gives the value of the shortest cycle containing node $v$.
(Here $\delta(u,v)$ refers to the weight of the shortest path from node $u$ to $v$)
This approach takes $O(APSP + n)$ rounds in total and is also used in Manoharan and Ramachandran~\cite{M22} to obtain
randomized algorithms for shortest cycle problems.
Since the current best bound for computing deterministic APSP is $\tilde{O}(n^{4/3})$~\cite{AR20},
this approach for computing shortest cycles requires 
$\tilde{O}(n^{4/3})$ rounds in total.

In Section~\ref{sec:framework}, we give a brief overview of our new framework
that involves computing a sequence of 
blocker sets, which then allows us to compute shortest cycles, 
without requiring an initial computation of APSP.
This, in effect, allows us to obtain $\tilde{O}(n)$ rounds algorithms for the shortest cycle problems.

\subsubsection{Undirected Shortest Cycles}  \label{sec:intro:undirected}
The problem of computing shortest cycles in undirected graphs is known to be much harder as 
compared to the directed setting. 
This has also been noted in an earlier work on MWC by Roditty and 
Williams~\cite{RW11}.
The reason behind this is that the reduction from the Shortest Cycle problems to APSP no longer works 
in the undirected setting:
an edge $(v,u)$ might also be the shortest path from $u$ to $v$, and $\min_{u} \{w(v,u) + \delta(u,v)\}$
might be $2\cdot w(v,u)$ and not the weight of the shortest cycle containing node $v$.

Instead in undirected graphs, the value of the minimum weight cycle can be obtained by taking the 
minimum across all possible values $\{w(v,u) + \delta(u,w) + \delta(v,w) \}$ 
where $u \in \mathcal{N}(v)$ and $w \in V$.
This requires an initial computation of APSP and $O(n)$ additional rounds for each node $v$ to 
communicate distance values $\delta(v,w)$, where $w \in V$, to all its neighbor nodes, thus
taking $O(APSP + n)$ rounds in total.
Since the current deterministic bound for APSP is $
\tilde{O}(n^{4/3})$~\cite{AR20}, this approach will also 
require $\tilde{O}(n^{4/3})$ for computing deterministic MWC.
We give a brief overview of our framework in Section~\ref{sec:framework} and a detailed description
of our $\tilde{O}(n)$ deterministic algorithm for computing undirected MWC in 
Section~\ref{sec:undirected}.

Our result provides a temporary polynomial gap between the 
current deterministic complexity of 
APSP and MWC computation in both directed
and undirected graphs.
All previously known deterministic results for these problems exhibit the 
same round complexity
of $\tilde{O}(n^{4/3})$~\cite{AR20}.

\subsubsection{Multi Source Shortest Paths (MSSP).}
The problem of computing MSSP
involves computing shortest path distances from all source 
nodes in a given source set $S$ to every node $v$ in the 
graph. 
The current approach to compute MSSP involves either computing
APSP or running SSSP using Bellman-Ford algorithm~\cite{B58}
for every node in the source set $S$.
This results in the total round bound of
$\tilde{O}(\min \{n^{4/3}, n\cdot |S|\})$.
In Section~\ref{sec:MSSP}, we describe our $\tilde{O}(n)$
rounds deterministic algorithm for computing MSSP, when there
are at most $\sqrt{n}$ source nodes.
Note that our algorithm works for both directed and undirected graphs.

\subsection{Our Framework for Computing Shortest Cycles} \label{sec:framework}

Our $\tilde{O}(n)$ deterministic algorithms for computing shortest cycles follows the following 
framework for a given graph $G$ (directed or undirected):

\begin{enumerate}
    \item Compute a series of blocker sets $Q_i$, for $1 \leq i \leq \lceil \frac{\log n}{\log \log n} \rceil - 1$, for $h_i$ hops, with source set $Q_{i-1}$, such that $Q_i$ intersects all exact shortest paths from source set $Q_{i-1}$ with hop-length $h_{i-1}$. Here $Q_{0} = V$, $h_{0} = \log^{2} n$ and $h_i = h_{i-1}\cdot \log n$.
    \item For each $i$, $1 \leq i \leq \lceil \frac{\log n}{\log \log n} \rceil - 1$, compute $h_{i}$-hop incoming shortest paths and $h_i$-hop outgoing shortest paths for source set $Q_i$. 
    \item Determine the weight of the shortest cycles using the values computed in Step (2).
\end{enumerate}

We now outline the ideas, followed by an informal description of each of the steps in our above
framework for computing shortest cycles.

\subsubsection{Deterministic Algorithm for Computing a Sequence of Blocker Sets (Step 1).}
For a given source set $S$ and hop-length $h$, 
Agarwal and Ramachandran~\cite{AR20} gave a deterministic algorithm that
computes a blocker set $Q$ for $S$ in $\tilde{O}(|S|\cdot h)$ rounds. 
Their algorithm involves framing the computation of a small blocker set as an approximate set cover
problem on a related hypergraph, and then adapting the efficient NC algorithm in Berger, Rompel and Shor~\cite{B94}
for computing an approximate minimum set cover problem in a hypergraph to an $\tilde{O}(|S|\cdot h)$
rounds randomized algorithm for computing the blocker set.
They then derandomize this algorithm to obtain an overall $\tilde{O}(|S|\cdot h)$ deterministic 
algorithm.

Our new contribution is to construct a sequence of blocker sets $Q_i$, 
for $1 \leq i \leq \lceil \frac{\log n}{\log \log n} \rceil - 1$,
such that nodes in $Q_i$ intersect all exact shortest paths from source nodes in $Q_{i-1}$
with hop-length $h_{i-1}$.
We use the deterministic blocker set algorithm in
Agarwal and Ramachandran~\cite{AR20} to compute these 
individual $Q_i$'s.
We describe this step in more detail in Section~\ref{sec:blocker} where we also establish 
additional properties that such a sequence of blocker sets guarantees, which then allows us to 
obtain our deterministic algorithms for computing shortest cycles.

\subsubsection{Computing Small-hop Shortest Paths from Blocker Sets \texorpdfstring{$Q_i$} (Step 2).}
After computing the sequence of blocker sets $Q_i$, for 
$1 \leq i \leq \lceil \frac{\log n}{\log \log n} \rceil - 1$, in Step (1), In Step (2) we compute
the $h_{i}$-hop outgoing and $h_{i}$-hop incoming shortest paths from 
source nodes in sets $Q_i$.
This takes $O(\sum_{i=0}^{\lceil \frac{\log n}{\log \log n} \rceil - 1} |Q_i|\cdot h_{i+1})$ rounds in total
using the Bellman-Ford algorithm~\cite{B58}.
We describe this further in Section~\ref{sec:prelim}.

\subsubsection{Computing the weight of Shortest Cycles (Step 3).}
Shortest cycles in both directed and undirected graphs exhibit the following
property: they are composed of one or more shortest paths.
In  directed graphs, the shortest cycle containing node $v$ consists of an
outgoing edge $(v,u)$, followed by the shortest path from $u$ to $v$.
Similarly, in undirected graphs, the shortest cycle for node $v$ consists of a 
shortest path from $v$ to some node $u$, followed by an edge $(u,w)$, and then
a shortest path from $w$ to node $v$.
We exploit this property of the shortest cycles, along with the small-hops
shortest path values obtained in Step 2, to obtain the weight of the shortest
cycles in both directed and undirected graphs.
We describe this in more detail in Sections~\ref{sec:directed} and
Section~\ref{sec:undirected}.

\subsection{Related Work}

\subsubsection{Deterministic Shortest Paths in CONGEST model.}
The classical Bellman-Ford algorithm~\cite{B58} is still the current best approach for
deterministically computing the single-source version of shortest paths (SSSP), even if 
the diameter of the graph is sublinear in the number of nodes and has
the round complexity of $O(n)$.

For the all-pairs version (APSP), the first non-trivial deterministic algorithm for this
problem was proposed by Agarwal et al.~\cite{ARKP18} and had the round bound complexity of
$\tilde{O}(n^{3/2})$.
This result was later improved by Agarwal and Ramachandran~\cite{AR20}, with a total
round complexity of $\tilde{O}(n^{4/3})$.

\subsubsection{Randomized Shortest Cycles in CONGEST model.}
A trivial approach to computing shortest cycles is to first 
compute APSP on the graph and then use the distance
values to compute the weight of the shortest cycles.
Such an approach has been used in
Manoharan and Ramachandran~\cite{M22}, 
where the authors use the 
randomized $\tilde{O}(n)$ rounds APSP algorithm of
Bernstein and Nanongkai~\cite{B21} to obtain 
randomized algorithms for shortest cycle problems.

One can argue whether it is possible to use such an approach to 
obtain $\tilde{O}(n)$ rounds deterministic algorithms for computing shortest cycles.
However, such an approach would require first using deterministic APSP to compute
the distance values, which are required to obtain the shortest cycles
algorithms in ~\cite{M22}.
Since the current best bound for deterministic APSP is 
$\tilde{O}(n^{4/3})$~\cite{AR20}, 
such an approach would 
require at least $\tilde{O}(n^{4/3})$ rounds.

\subsection{Organization of the Paper}

In Section~\ref{sec:prelim} we give an overview of the well-known CONGEST model
and the Bellman-Ford algorithm for computing small hop shortest paths.
We describe our deterministic algorithm for computing a sequence of 
successive blocker sets in Section~\ref{sec:blocker}.
In Section~\ref{sec:directed}, we describe our $\tilde{O}(n)$ deterministic 
algorithm for computing
minimum weight cycle in directed graphs, and in Section~\ref{sec:undirected}
we describe our $\tilde{O}(n)$ deterministic 
algorithm for computing
minimum weight cycle in undirected graphs.
Finally in Section~\ref{sec:MSSP}, we describe our 
$\tilde{O}(n)$ rounds deterministic algorithm for computing 
multi-source shortest paths.

\section{Preliminaries} \label{sec:prelim}

\subsection{CONGEST model}  

We now describe the well-known CONGEST model for distributed computation.
In this model, $n$ independent processors are interconnected in a network
by bounded-bandwidth links. 
We refer to these processors as nodes and the links as edges. 
This network is modeled by a graph $G = (V, E)$ where $V$ refer to the set of 
processors and $E$ refer to the set of links between the processors. 
The model also assumes that the system is fault-free, i.e. nodes do not crash or
behave badly and messages are not lost, delayed, or corrupted.

The processors (or nodes) are assumed to have unique IDs in the range of
$\{1, \ldots, poly(n) \}$ and have infinite computational power available
on each node.
Each node has limited topological knowledge and only knows about its incident
edges.
Each edge has an arbitrary real non-negative edge weight.
The communication network of $G$ is represented by the underlying undirected
graph $U_G$ of $G$.
The communication channels across all edges are bidirectional, even if the 
edges are directed.
Note that the model allows a node to send different messages along different edges.

The computation in this model proceeds in synchronous rounds.
In each round,
each node $v \in V$ can send a constant number of words along each outgoing edge, and it 
receives the messages sent to it in the previous round.
The CONGEST model normally assumes that a word has $O(\log n)$ bits.
Since we allow arbitrary edge-weights, here we assume that a constant number of node-ids, 
edge-weights, and distance values can be sent along every edge in every round (similar assumptions
are made in~\cite{AR20, ARKP18, B21, E20}).
After receiving
the messages, every node performs some local computation.
We measure the performance of algorithms in terms of their round complexity, 
defined as the worst-case number of rounds of distributed communication.

\subsection{Notations}

Let $G = (V, E)$ be a directed or undirected graph with arbitrary non-negative 
weights: $V$ is the set of nodes, and $E$ is the set of edges, such that 
$|V| = n$ and $|E| = m$. 
Let $\mathcal{N}_{out}(v)$ refer to the set of out neighbors of 
node $v$ and let $\mathcal{N}_{in}(v)$ refer to the set of in
neighbors of node $v$.
In undirected graphs, we simply denote the neighbors of node $v$
with $\mathcal{N}(v)$, since edges are not directed.

Let $w(u, v)$ denote the weight of the edge $(u,v)$
and let $\delta(s,t)$ refer to the weight of the shortest path 
from node $s \in V$ to $t \in V$.
For a given hop-length $h$, the value $\delta^{h}(s,t)$ refer to the weight
of the shortest path from $s$ to $t$ among all such paths from $s$ to $t$, that
has at most $h$ hops.
Let $\delta (v)$ denote the weight of the shortest cycle containing node $v$ and
let $\delta^{h}(v)$ denote the weight of the shortest cycle containing node
$v$, that has at most $h$ hops.
Let $wt(\mathcal{P})$ refer to the weight of a path $\mathcal{P}$ and let 
$wt(\mathcal{C})$ refer to the weight of a cycle $\mathcal{C}$.

Let $parent(s,t)$ refer to the 
parent of node $t$ on the computed shortest 
path between nodes $s$ and $t$ and let $parent^{h}(s,t)$ refer to 
the parent of node $t$ on the computed $h$-hop shortest 
path between nodes $s$ and $t$.
Similarly let $next(s,t)$ refer to the 
node next to $s$, which lies on the computed shortest 
path between nodes $s$ and $t$ and let $next^{h}(s,t)$ refer to 
the node next to $s$, which lies on the computed $h$-hop shortest 
path between nodes $s$ and $t$.

\subsection{Small Hops Shortest Paths}  \label{sec:small}

We now define the notion of $h$-hops shortest paths and show how traditional
Bellman-Ford algorithm~\cite{B58} is used to compute these $h$-hops 
shortest paths.
We start with defining the notion of $h$-hop-accurate shortest path distances
defined in~\cite{B21}.

\begin{definition}[$h$-hop accurate~\cite{B21}]
For an integer $h$ and a source node $s \in V$ and a target node $t \in V$, a distance value $\delta^{h}(s,t)$ is $h$-hop accurate if the following holds: (1) $\delta^{h}(s,t) \geq \delta(s,t)$ and, (2) if there is a shortest path from $s$ to $t$ of hop-length $\leq h$, then $\delta^{h}(s,t) = \delta(s,t)$.
\end{definition}

\subsubsection{Bellman Ford Algorithm.}
We now show how the Bellman-Ford algorithm computes these $h$-hop accurate shortest
path distance values from a given source node $s$.
The algorithm runs for $h$ rounds, where $h$ is an input given by the user.
Here is a brief description of the algorithm.

\begin{algorithm}
\caption{Computing $h$-hop outgoing shortest path distances using Bellman-Ford algorithm~\cite{B58}}
\begin{algorithmic}[1]
\Statex \textbf{Input:} $G = (V, E)$; \hspace{5pt} source : $s$; \hspace{5pt} number of hops: $h$
\State \textbf{For $\bm{0\leq i \leq h}$:} Initialize $D^{i}(s,t) = \infty$ for $t \neq s$; \hspace{5pt} Set $D^{i}(s,s) = 0$.
\State \textbf{For $\bm{0\leq i \leq h}$:} Initialize $parent^{h}(s,t) \leftarrow t$ for $t \neq s$; \hspace{5pt} Set $parent^{h}(s,s) \leftarrow s$
\State \textbf{Round $\bm{i}$:} ($0\leq i \leq h-1$)
\State \hspace{10pt} Every node $v$ sends the computed distance value $D^{i}(s,v)$ to all its outgoing neighbors.    \label{algBF:send}
\State \hspace{10pt} $\begin{aligned}[t] 
    & \text{Using the values received in Step~\ref{algBF:send}, node $v$ computes the distance value} \\
    & \text{$D^{i+1}(s,v) \leftarrow \min_{u \in \mathcal{N}_{in}(v)} \{D^{i}(s,u) + w(u,v)\}$.} 
    \end{aligned}$
\State \hspace{10pt} \textbf{if $D^{i+1}(s,v) \neq D^{i}(s,v)$ then:} set $parent^{h}(s,v) \leftarrow \arg \min_{u \in \mathcal{N}_{in}(v)} \{D^{i}(s,u) + w(u,v)\}$
\end{algorithmic}  \label{algBF}
\end{algorithm} 

Clearly, Algorithm~\ref{algBF} takes $O(h)$ rounds.
Moreover, it can be proved that when the algorithm terminates the computed
$h$-hop
shortest path distance values, $D^h (s,v)$, is equal to  $\delta^{h}(s,v)$, and
is $h$-hop-accurate for every 
node $v \in V$.
Note that we can modify the above algorithm if we instead want to compute
$h$-hop incoming shortest path distance values
$\delta^{h}(v,s)$ at each node $v \in V$.
We describe this modified algorithm
in Algorithm~\ref{algBFIncoming}, that computes 
$\delta^{h}(v,s)$ distance values at each node $v \in V$, and also takes 
$O(h)$ rounds.

\begin{algorithm}
\caption{Computing $h$-hop incoming shortest path distances using Bellman-Ford algorithm~\cite{B58}}
\begin{algorithmic}[1]
\Statex \textbf{Input:} $G = (V, E)$; \hspace{5pt} source : $s$; \hspace{5pt} number of hops: $h$
\State \textbf{For $\bm{0\leq i \leq h}$:} Initialize $D^{i}(t,s) = \infty$ for $t \neq s$; \hspace{5pt}  Set $D^{i}(s,s) = 0$.
\State \textbf{For $\bm{0\leq i \leq h}$:} Initialize $next^{h}(t,s) \leftarrow t$ for $t \neq s$;  \hspace{5pt} Set $next^{h}(s,s) \leftarrow s$
\State \textbf{Round $\bm{i}$:} ($0\leq i \leq h-1$)
\State \hspace{10pt} Every node $v$ sends the distance value $D^{i}(v,s)$ to all its incoming neighbors.    \label{algBFIncoming:send}
\State \hspace{10pt} $\begin{aligned}[t] 
    & \text{Using the values received in Step~\ref{algBF:send}, node $v$ computes the distance value} \\
    & \text{$D^{i+1}(v,s) \leftarrow \min_{u \in \mathcal{N}_{out}(v)} \{w(v,u) + D^{i}(u,s)\}$.} 
    \end{aligned}$
\State \hspace{10pt} \textbf{if $D^{i+1}(v,s) \neq D^{i}(v,s)$ then:} set $next^{h}(v,s) \leftarrow \arg \min_{u \in \mathcal{N}_{out}(v)} \{w(v,u) + D^{i}(u,s)\}$
\end{algorithmic}  \label{algBFIncoming}
\end{algorithm} 

\subsection{Blocker Sets}

We now describe the notion of blocker sets from
King~\cite{K99} where it was used 
in the dynamic setting to compute fully dynamic APSP.
This technique was then first adapted to the distributed CONGEST model 
in~\cite{ARKP18} and later improved to run with almost optimal bounds 
in~\cite{AR20}.
The blocker sets, in essence, is a deterministic tool to compute Set Cover 
deterministically, when the underlying set is a collection of paths. 
We now define the notion of \textit{$h$-hop Consistent SSSP Collection}
(or $h$-CSSSP) which is used in~\cite{AR20} to compute blocker sets
efficiently.

\begin{definition}[$h$-hop CSSSP~\cite{A19}]
Let $\mathcal{T}$ be a collection of trees of depth $h$ for a set of sources 
$S \subseteq V$. Then $\mathcal{T}$ is an $h$-hop CSSSP collection if for 
every $u,v \in V$ the path from $u$ to $v$ is the same in each of the trees in 
$\mathcal{T}$ (if such a path exists), including 
the $h$-hop tree $T_u$ rooted at $u$. 
Further, each $T_u$ contains every vertex $v$ that has a path with at most $h$ 
hops from $u$ in $G$ that has shortest path distance $\delta(u, v)$.
\end{definition}

Using the Bellman-Ford algorithm described in Section~\ref{sec:small},
one can compute $h$-hop CSSSP for a source set $S$ and hop-length $h$ in
$O(|S|\cdot h)$ rounds.
The advantage of using $h$-CSSSP instead of a collection of $h$-hop shortest
paths is that it guarantees that all paths rooted at any vertex $v$ 
constitutes a rooted tree at $v$, which is not guaranteed if we only use
the collection of $h$-hop shortest paths.
This is described in detail in~\cite{A19}.

\begin{lemma}[~\cite{A19, B58}] \label{lemma:hCSSSP}
    For a source set $S$ and hop-length $h$, $h$-hop CSSSP can be computed in $O(|S|\cdot h)$ rounds.
\end{lemma}

We now define the notion of blocker set from
King~\cite{K99} which was later adapted 
to the CONGEST model in~\cite{ARKP18}.

\begin{definition}[Blocker Set~\cite{K99, ARKP18}]  \label{def:blocker}
    Let $H$ be a collection of rooted $h$-hop trees for a source set 
    $S \subseteq V$. A set $Q \subseteq V$ is a blocker set for paths in the set $H$ if every root-to-leaf path of length $h$ in every tree in $H$ contains a node in set $Q$.
\end{definition}

We use the following result from Agarwal and Ramachandran~\cite{AR20} which we use in our 
algorithm to compute a desired sequence of 
successive blocker sets described in Section~\ref{sec:blocker}.

\begin{lemma}[~\cite{AR20}]
    For a source set $S \subseteq V$ and hop-length $h$ and given its 
    corresponding $h$-CSSSP collection $\mathcal{T}$, a blocker set $Q$
    for the collection of paths in $\mathcal{T}$ can be computed 
    deterministically in $\tilde{O}(|S|\cdot h/(\epsilon^2\delta^3))$
    rounds, where both $\epsilon$ and $\delta$ are small positive constants 
    with a value less than or equal to $1/12$.
        Also, the blocker set $Q$ has size $O(n\log n /h)$.
\end{lemma}

Since $\epsilon$ and $\delta$ are small positive constants and the only
variables are $S$ and $h$, we can deduce the following simpler corollary
from the above lemma.

\begin{corollary}[~\cite{AR20}] \label{corollary:Q}
    For a source set $S \subseteq V$ and hop-length $h$ and given its 
    corresponding $h$-CSSSP collection $\mathcal{T}$, a blocker set $Q$
    for the collection of paths in $\mathcal{T}$ can be computed 
    deterministically in $\tilde{O}(|S|\cdot h)$ rounds.
    Also, the blocker set $Q$ has size $O(n\log n /h)$.
\end{corollary}

\subsection{Broadcasting Primitives}   

We now describe a commonly known broadcasting primitive that we use in 
many of our algorithms.

\begin{lemma}[~\cite{P00}]  \label{lemma:broadcast}
Let $G = (V, E)$ be a directed or undirected graph.
Suppose each node $v \in V$ holds $k_v \geq 0$ messages of $O(\log n)$
bits each, for a total of $K = \sum_{v\in V} k_v$ messages.
Then all nodes in the network can receive
these $K$ messages within $O(K + n)$ rounds, where $n = |V|$.
\end{lemma}

See Peleg~\cite{P00} for the proof of 
Lemma~\ref{lemma:broadcast}.

\section{Computing Sequence of Blocker Sets}    \label{sec:blocker}

In this Section, we describe our algorithm for computing a sequence of
successive blocker sets.
A similar approach has been used earlier in the randomized setting in
the distributed computation of 
randomized all pairs shortest paths (APSP) by Bernstein and 
Nanongkai~\cite{B21}.
We use the deterministic blocker set algorithm of
Agarwal and Ramachandran~\cite{AR20} to compute
the individual blocker sets in our deterministic algorithm 
for computing this sequence of successive blocker sets.

Our sequence of successive blocker sets $\{Q_i\}$ exhibit very strong
properties related to long-hop shortest paths, which allows us to obtain
our $\tilde{O}(n)$ rounds deterministic algorithms for minimum weight cycle
and multi-source shortest paths
problems described in Sections~\ref{sec:directed}, 
\ref{sec:undirected} 
and \ref{sec:MSSP}.
One such crucial property is that given two nodes $s$ and $t$ such that the true
shortest path from $s$ to $t$ has very high hop length, then there exists a 
blocker node $q$ in one of the blocker sets in the sequence $\{Q_i\}$ such
that the long-hop shortest path from $s$ to $t$ can be composed of a 
short-hop incoming shortest path from $q$ to $s$ and a 
short-hop outgoing
shortest path from $q$ to $t$.
We establish this property for the sequence of blocker sets $\{Q_i\}$
later on in this section.

\begin{algorithm}
\caption{Computing Sequence of Blocker Sets $\{Q_i\}$}
\begin{algorithmic}[1]
\Statex \textbf{Input:} $G = (V, E)$; 
\State Set $Q_0 \leftarrow V$; $h_0 \leftarrow \log^2 n$  \label{algBlocker:set}
\For{$1 \leq i \leq \lceil \frac{\log n}{\log \log n} \rceil - 1$}  \label{algBlocker:startFor}
\State Compute $h_{i-1}$-CSSSP $\mathcal{T}_{i-1}$ for source set $Q_{i-1}$ using the algorithm in~\cite{A19}.  \label{algBlocker:CSSSP} 
\State $\begin{aligned}[t] 
    & \text{Compute blocker set $Q_i$ for the paths in $\mathcal{T}_{i-1}$ with $Q_{i-1}$ as the source set, using} \\
    & \text{the deterministic blocker set algorithm in~\cite{AR20}.}
    \end{aligned}$ \label{algBlocker:blocker}
\State Every node $q_i \in Q_i$ broadcast its ID to the entire network. \label{algBlocker:broadcast}
\State $h_i \leftarrow h_{i-1}\cdot \log n$ \label{algBlocker:update}
\EndFor \label{algBlocker:endFor}
\end{algorithmic}  \label{algBlocker}
\end{algorithm} 

Algorithm~\ref{algBlocker} describes our deterministic algorithm for computing 
a sequence of blocker sets in a given graph $G = (V,E)$.
In Step~\ref{algBlocker:set} we set the initial blocker set $Q_0$ to the 
entire vertex set $V$ and we set the initial hop-length $h_0$ to $\log^2 n$.
The for loop in Steps~\ref{algBlocker:startFor}-\ref{algBlocker:endFor}
runs for $\lceil \frac{\log n}{\log \log n} \rceil - 1$ iterations in total.

Each iteration $i$ of the for loop proceeds in the following manner:
first in Step~\ref{algBlocker:CSSSP} an $h_{i-1}$-CSSSP $\mathcal{T}_{i-1}$
is constructed for the source set $Q_{i-1}$.
We use the h-CSSSP algorithm in~\cite{A19} to construct this CSSSP.
Then in Step~\ref{algBlocker:blocker} we compute a blocker set $Q_i$ for the
paths in the collection $\mathcal{T}_{i-1}$ using the blocker set algorithm 
in~\cite{AR20}.
The nodes that are selected in the set $Q_i$ broadcast their IDs to the entire
network in Step~\ref{algBlocker:broadcast} in $O(|Q_i| +n )$ rounds.
We then update the hop-length $h_i$ for the next iteration in 
Step~\ref{algBlocker:update}.

We now establish some important properties exhibited by the sequence of
blocker sets constructed in Algorithm~\ref{algBlocker}.
The following lemma shows that for every pair of nodes $s$ and $t$ for which 
there exists a shortest path between them,
there exists a blocker set $Q_j$ in
the sequence $\{Q_i\}$ such that for some node $q_j \in Q_j$,
the weight of the shortest path from $s$ to $t$, $\delta(s,t)$, is equal to the 
sum of the $h_j$-hop 
shortest path distances $\delta^{h_j}(s,q_j)$ and $\delta^{h_j}(q_j,t)$.

\begin{lemma}   \label{lemma:property:Qj}
Let $G = (V, E)$ be a directed or undirected graph and let $\{Q_i\}$ be the 
sequence of blocker sets computed in Algorithm~\ref{algBlocker} for graph $G$.
Then for each node pairs $s, t \in V$ for which a shortest path exists, 
there exists some $j$ such that for some node $q_j \in Q_j$,
the weight of the shortest path from $s$ to $t$, 
$\delta(s,t) = \delta^{h_j}(s,q_j) + \delta^{h_j}(q_j, t)$.
\end{lemma}

\begin{proof}
    Consider node pair $(s,t)$ in $G$ such that 
    among all node pairs for which the above lemma does not 
    hold, the value of the shortest path distance $\delta(s,t)$ is the 
    minimum among them.
    If there are multiple such pairs, pick the one with the minimum 
    shortest path hop-length.
    Now we need to show that no such pair exists in $G$.

    Let $p_{s,t}$ be the shortest path from $s$ to $t$.
    If there are multiple such paths, pick the one with the minimum hop length.
    Let $(a,t)$ be the last edge on this path.
    If $a = s$ and the hop length is 1, then the above lemma holds 
    as $s \in Q_0$ and $h_0 \geq 1$.

    We now consider the case when $a \neq s$.
    Let the subpath from $s$ to $a$ be $p_{s,a}$.
    Since $p_{s,a}$ has hop-length less than $p_{s,t}$ and has weight 
    less than or equal to $\delta(s,t)$, the above lemma holds for pair 
    $(s,a)$.
    Thus there exists some $j'$ and a node $q_{j'} \in Q_{j'}$
    such that 
    $\delta(s, a) = \delta^{h_{j'}}(s, q_{j'}) + \delta^{h_{j'}}(q_{j'}, a)$.

    Now since the above lemma does not hold for pair $(s,t)$, it implies that
    the path from $q_{j'}$ to $t$ has hop-length strictly greater than $h_{j'}$
    and the hop-length from $q_{j'}$ to $a$ is equal to $h_{j'}$.
    Since the hop-length from $q_{j'}$ to $a$ is equal to $h_{j'}$, by the 
    definition of blocker set (Definition~\ref{def:blocker}), there exists 
    some node $q_{j' + 1} \in Q_{j' + 1}$, such that $q_{j' + 1}$ lies on the 
    path from $q_{j'}$ to $a$. 
    Now we can re-write the distance value $\delta(s,t)$ as:

    \begin{align*}
    \delta(s,t) & = \delta(s,a) + w(a,t)    \\
    & = \delta^{h_{j'}}(s, q_{j'}) + \delta^{h_{j'}}(q_{j'}, a) + w(a,t)    \\
    & = \delta^{h_{j'}}(s, q_{j'}) + \delta^{h_{j'}}(q_{j'}, q_{j'+1}) + 
    \delta^{h_{j'}}(q_{j'+1}, a) + w(a,t) \\
    & \geq \delta^{2\cdot h_{j'}}(s, q_{j'+1}) + \delta^{1+h_{j'}}(q_{j'+1}, t) \\
    & \geq \delta^{h_{j'+1}}(s, q_{j'+1}) + \delta^{h_{j'+1}}(q_{j'+1}, t) \text{\hspace{.1in} (since $2\cdot h_{j'} \leq \log n \cdot h_{j'} = h_{j'+1}$)} 
    \end{align*}

    Since the sum of the distance values $\delta^{h_{j'+1}}(s, q_{j'+1})$, and $\delta^{h_{j'+1}}(q_{j'+1}, t)$ corresponds to a path from $s$ to $t$, hence:
    \begin{align*}
    \delta(s,t) & = \delta^{h_{j'+1}}(s, q_{j'+1}) + \delta^{h_{j'+1}}(q_{j'+1}, t) 
    \end{align*}

    The above equation shows that the property holds for pair $(s,t)$
    for $j = j' +1$, contradicting our earlier assumption.
    This shows that the stated property holds for all node pairs $(s,t)$
    such that a shortest path exists between them.
    
\end{proof}

We now provide an upper bound on the size of the blocker sets $Q_j$ in the 
blocker set sequence $\{Q_i\}$ (Lemma~\ref{lemma:Qj}) and then
using this
we show that Algorithm~\ref{algBlocker}
runs in $\tilde{O}(n)$ rounds (Lemma~\ref{lemma:algBlocker:bound}).

\begin{lemma}   \label{lemma:Qj}
    Let $G = (V, E)$ be a directed or undirected graph and let $\{Q_i\}$ be the 
sequence of blocker sets computed in Algorithm~\ref{algBlocker} for graph $G$.
Then for each $1 \leq j \leq \lceil \frac{\log n}{\log \log n} \rceil - 1$,
the blocker set $Q_j$ in the blocker set sequence $\{Q_i\}$ has size 
$O(\frac{n\log n}{h_{j-1}})$.
\end{lemma}

\begin{proof}
    In Step~\ref{algBlocker:blocker}, the blocker set $Q_j$ is constructed for
    the paths in collection $\mathcal{T}_{j-1}$, which have hop-length at most
    $h_{j-1}$.
    Thus by Corollary~\ref{corollary:Q}, the set $Q_j$ has size
    $O(\frac{n\log n}{h_{j-1}})$.
\end{proof}

\begin{lemma}   \label{lemma:algBlocker:bound}
    Algorithm~\ref{algBlocker} computes the blocker set sequence $\{Q_i\}$ 
    deterministically in $\tilde{O}(n)$ rounds.
\end{lemma}

\begin{proof}
    We first show that each iteration of the for loop in 
    Steps~\ref{algBlocker:startFor}-\ref{algBlocker:endFor} takes $\tilde{O}(n)$
    rounds.
    Using Lemma~\ref{lemma:hCSSSP}, Step~\ref{algBlocker:CSSSP} takes
    $O(|Q_{i-1}|\cdot h_{i-1}) = O(\frac{n\log n}{h_{i-2}} \cdot h_{i-1})$ 
    (since $|Q_{i-1}| = O(\frac{n\log n}{h_{i-2}})$ using Lemma~\ref{lemma:Qj}) 
    rounds. 
    Since $h_{i-1} = h_{i-2}\cdot \log n$, Step~\ref{algBlocker:CSSSP} takes
    $O(n\log^2 n)$ rounds in each iteration $i$.
    The computation of blocker set $Q_i$ in Step~\ref{algBlocker:blocker} 
    takes $\tilde{O}(|Q_{i-1}|\cdot h_{i-1})$ rounds 
    (Corollary~\ref{corollary:Q}).
    As earlier, 
    $\tilde{O}(|Q_{i-1}|\cdot h_{i-1}) = \tilde{O}(n\log^2 n) = \tilde{O}(n)$.
    Step~\ref{algBlocker:broadcast} takes $O(n)$ rounds 
    (Lemma~\ref{lemma:broadcast}).
    Since there are in total $O(\frac{\log n}{\log \log n})$ iterations of the 
    for loop, the entire execution of Algorithm~\ref{algBlocker} takes
    $\tilde{O}(n)$ rounds.
\end{proof}

\section{Directed Shortest Cycles}    \label{sec:directed}

In this section, we describe our deterministic algorithm for computing 
the weight of shortest cycles in a directed graph $G = (V, E)$.
A trivial way to compute these shortest cycles is first to do an initial 
computation of APSP on the graph and then using the 
computed shortest path distance values to compute the weight of the 
shortest cycles.
Specifically for each node $v \in V$, the weight of the shortest cycle containing
node $v$ can be calculated by taking the minimum of $\{w(v,u) + \delta(u, v)\}$
across all outgoing neighbors $u$ of node $v$.
This approach results in a $\tilde{O}(n^{4/3})$ round bound as the current
best deterministic bound for APSP is $\tilde{O}(n^{4/3})$~\cite{AR20}.

Our algorithm for computing shortest cycles differs from the above described
approach in that it does not need an initial computation of APSP.
Instead we first compute the blocker set sequence $\{Q_i\}$ for $G$ and then
compute the small hop shortest path distances from each of the blocker sets
$Q_j$ in this sequence.
We then use these small hop distance values to compute the weight of the
shortest cycles for each of the nodes $v$ in graph $G$.
This approach, in effect, allows us to obtain a $\tilde{O}(n)$ rounds
deterministic algorithm for the shortest cycles problem.

The blocker set sequence $\{Q_i\}$ described in Section~\ref{sec:blocker}
exhibits another strong property in the context of shortest cycles which is crucial
in obtaining our shortest cycles algorithm.
Specifically given a node $v \in V$ such that the shortest cycle containing 
node $v$ has a very high hop length, then there exists a blocker node $q$ in one 
of the blocker sets in the sequence $\{Q_i\}$ such that the shortest cycle at 
node $v$ consists of an outgoing edge $(v,u)$, followed by a short-hop 
incoming shortest path from source $q$ to $u$ and a short-hop outgoing shortest
path from source $q$ to node $v$.
We state this formally in Lemma~\ref{property-cycle-directed}.

\begin{lemma}   \label{property-cycle-directed}
    Let $G = (V,E)$ be a directed graph and let $\{Q_i\}$ be the sequence of
    blocker sets computed in Algorithm~\ref{algBlocker} for graph $G$.
    Let $C_v$ denote the shortest cycle containing node $v$, if such a cycle 
    exists.
    Then there exists $u \in \mathcal{N}(v)$ and a $j$ such that for some
    node $q_j \in Q_j$, the weight of cycle $C_v$, 
    $wt(C_v) = w(v,u) + \delta^{h_j}(u,q_j) + \delta^{h_j}(q_j, v)$.
\end{lemma}

\begin{proof}
    Let $(v,u)$ be the edge incident to $v$ on cycle $C_v$. 
    Let $p_{u,v}$ be the subpath from $u$ to $v$ in $C_v$.
    Clearly $p_{u,v}$ is a shortest path from $u$ to $v$ in $G$, otherwise
    $C_v$ would not be the shortest cycle for node $v$.
    Then using Lemma~\ref{lemma:property:Qj}, there exists $q_j \in Q_j$
    such that $\delta(u,v)$ is equal to the sum of values $\delta^{h_j}(u,q_j)$
    and $\delta^{h_j}(q_j,v)$.
    This establishes the above lemma.
\end{proof}

Algorithm~\ref{algDirected} describes our deterministic algorithm for computing 
the shortest cycles for each node $v \in V$ in a given  directed graph 
$G = (V,E)$.
In Step~\ref{algDirected:set} we set the initial hop length $h_0$ to $\log^2 n$
and we also set the shortest cycle distance values, $\delta (v)$, 
for all nodes $v \in V$ to $\infty$.
We then compute the deterministic blocker set sequence $\{Q_i\}$ for graph $G$
using the algorithm described in Section~\ref{sec:blocker} in 
Step~\ref{algDirected:blocker}.
The for loop in Steps~\ref{algDirected:startFor}-\ref{algDirected:endFor}
runs for $\lceil \frac{\log n}{\log \log n} \rceil$ iterations in total.

\begin{algorithm}
\caption{Computing Directed Shortest Cycles}
\begin{algorithmic}[1]
\Statex \textbf{Input:} Directed Graph $G = (V, E)$
\State Set $h_0 \leftarrow \log^2 n$; \hspace{5pt} $\delta (v) \leftarrow \infty$ for all $v \in V$  \label{algDirected:set}
\State Compute deterministic blocker set sequence $\{Q_i\}$ for $G$ using Algorithm~\ref{algBlocker}. \label{algDirected:blocker}
\For{$0 \leq i \leq \lceil \frac{\log n}{\log \log n} \rceil - 1$}  \label{algDirected:startFor}
\State $\begin{aligned}[t]
    & \text{Compute $h_i$-hop in-SSSP for all source nodes $q_i \in Q_i$ using Bellman-Ford} \\
    & \text{algorithm~\cite{B58} (Algorithm~\ref{algBFIncoming}).}
    \end{aligned}$ \label{algDirected:inHop}
\State  $\begin{aligned}[t]
    & \text{Compute $h_i$-hop out-SSSP for all source nodes $q_i \in Q_i$ using Bellman-Ford} \\
    & \text{algorithm~\cite{B58} (Algorithm~\ref{algBF}).}
    \end{aligned}$ \label{algDirected:outHop}
\State $\begin{aligned}[t]
    & \text{Each node $u \in V$ sends to all its incoming neighbors $v \in \mathcal{N}_{in}(u)$, the $h_i$-hop in-SSSP} \\
    & \text{distance values, $\delta^{h_i}(u,q_i)$, for all source nodes $q_i \in Q_i$.}
    \end{aligned}$ \label{algDirected:send}
\State \textbf{Local Step at each node $\bm{v \in V}$:} \label{algDirected:local}
\State \hspace{10pt} $\delta (v) \leftarrow \min \{\delta (v), \text{\hspace{5pt}}\min_{q_i \in Q_i} \{ \min_{u \in \mathcal{N}_{out}(v)} (w(v,u) + \delta^{h_i}(u,q_i) + \delta^{h_i}(q_i,v)) \} \} $ \label{algDirected:deltav}
\State $h_i \leftarrow h_{i-1}\cdot \log n$ \label{algDirected:update}
\EndFor \label{algDirected:endFor}
\end{algorithmic}  \label{algDirected}
\end{algorithm} 

Each iteration $i$ of the for loop in 
Steps~\ref{algDirected:startFor}-\ref{algDirected:endFor} proceeds in the 
following manner: 
Step~\ref{algDirected:inHop} computes the $h_i$-hop incoming shortest paths
from all source nodes $q_i$ in the blocker set $Q_i$ using the well-known
Bellman-Ford algorithm~\cite{B58} (also described in 
Algorithm~\ref{algBFIncoming}).
Similarly in Step~\ref{algDirected:outHop} using Bellman-Ford 
algorithm~\cite{B58} (Algorithm~\ref{algBF}),
we compute the $h_i$-hop outgoing
shortest paths from all source nodes $q_i$ in the blocker set $Q_i$.
Then in Step~\ref{algDirected:send} every node $u \in V$ sends all the computed
$h_i$-hop incoming shortest path distance values, $\delta^{h_i} (u,q_i)$,
for all blocker nodes $q_i \in Q_i$ to all its incoming neighbor nodes 
$v \in \mathcal{N}_{in}(u)$.

Using the received values in Step~\ref{algDirected:send}, every node $v \in V$ then computes the weight
of candidate shortest cycles passing through it, by summing up the 
following three components: (1) weight of an outgoing edge $(v,u)$, 
(2) followed by the weight of the $h_i$-hop shortest path from $u$ to 
a blocker node $q_i \in Q_i$, $\delta^{h_i} (u,q_i)$, 
(3) and the weight of the $h_i$-hop shortest path from  $q_i$ to $v$, 
$\delta^{h_i} (q_i, v)$.
Node $v$ then calculates the weight of all such candidate cycles for all 
of its outgoing neighbor nodes $u$ and for all blocker nodes $q_i \in Q_i$.
It then updates the value of $\delta (v)$, if the weight of one of the 
candidate cycles is smaller than the current value of $\delta (v)$
(Steps~\ref{algDirected:local}-\ref{algDirected:deltav}).

We now show that, for each node $v \in V$, Algorithm~\ref{algDirected}
correctly computes the weight of the shortest cycle passing through node $v$,
if such a cycle exists in the graph $G$.

\begin{lemma}   \label{lemma:algDirected:correctness}
    Let $G = (V,E)$ be a directed graph.
    Then for each node $v \in V$, Algorithm~\ref{algDirected} correctly
    computes the weight of the shortest cycle passing through node $v$, 
    $\delta(v)$, if such a cycle exists.
\end{lemma}

\begin{proof}
    Fix a node $v$.
    We consider two cases: (i) the case when $v$ is not part of 
    any cycle in graph $G$, (ii) the case when $v$ is part of 
    some cycle in graph $G$.
    \newline

    \textit{(i) $v$ is not part of any cycle in graph $G$:}
    In this case, we need to show that the shortest cycle distance
    value for node $v$, $\delta (v)$, remains at $\infty$ 
    throughout the execution of the algorithm.
    On the contrary assume that $\delta (v)$ was instead set to 
    some finite value.
    The only step where the $\delta (v)$ values are 
    updated is Step~\ref{algDirected:deltav}, where it is updated 
    to value $w(v,u) + \delta^{h_i}(u,q_i) + \delta^{h_i}(q_i, v)$,
    for some outgoing neighbor $u$ of $v$ and a blocker node
    $q_i \in Q_i$.
    It further implies that both distance values, 
    $\delta^{h_i}(u,q_i)$ and $\delta^{h_i}(q_i, v)$, are finite.
    Let $p_{u,q_i}$ be the path that corresponds to the 
    distance value $\delta^{h_i}(u,q_i)$ and let path $p_{q_i, v}$
    corresponds to the distance value $\delta^{h_i}(q_i, v)$.
    But the edge $(v,u)$, followed by paths $p_{u,q_i}$ and
    $p_{q_i, v}$ gives us a cycle involving node $v$, resulting in
    a contradiction.
    \newline
    
    \textit{(ii) $v$ is part of some cycle in $G$:}
    Let $C_v$ be a shortest cycle for node $v$.
    In this case, we need to show that the value for $\delta (v)$
    computed by the algorithm equals the weight of $C_v$.
    By Lemma~\ref{property-cycle-directed}, there exists 
    $u \in \mathcal{N}_{out}(v)$ and a $j$ such that for some blocker 
    node $q_j \in Q_j$:
    $$wt(C_v) = w(v,u) + \delta^{h_j}(u,q_j) + \delta^{h_j}(q_j,v)$$
    But then this value would have been captured in 
    Step~\ref{algDirected:deltav} of the Algorithm, thus
    $\delta(v) \leq wt(C_v)$.

    Now if the value of $\delta(v)$ is strictly less than $wt(C_v)$,
    then similar to the previous case (Case (i)), we can argue
    that there exists another cycle composed of an edge 
    $(v, u')$, followed by a path from $u'$ to $q_{i^{'}}$ of 
    hop-length at most $h_{i'}$, and another path from $q_{i^{'}}$
    to $v$ of hop-length at most $h_{i'}$ and has total weight
    less than $wt(C_v)$, thus resulting in a contradiction as 
    $C_v$ is a shortest cycle through node $v$.
    This completes the proof.
\end{proof}

In the following lemma, we show that Algorithm~\ref{algDirected}
runs in $\tilde{O}(n)$ rounds.

\begin{lemma}   \label{lemma:algDirected:bound}
Algorithm~\ref{algDirected} computes all $\delta(v)$ values 
deterministically in $\tilde{O}(n)$ rounds.
\end{lemma}

\begin{proof}
    Step~\ref{algDirected:set} is a local step.
    Using Lemma~\ref{lemma:algBlocker:bound}, the blocker set
    sequence $\{Q_i\}$ can be computed deterministically in 
    $\tilde{O}(n)$ rounds.
    We now show that each iteration of the for loop in 
    Steps~\ref{algDirected:startFor}-\ref{algDirected:endFor}
    takes $\tilde{O}(n)$ rounds.

    Using Bellman-Ford algorithm~\cite{B58}, 
    Steps~\ref{algDirected:inHop} \& \ref{algDirected:outHop} takes 
    $O(|Q_i|\cdot h_i) = O(\frac{n\log n}{h_{i-1}}\cdot h_i) = O(n\log^2 n)$ (since $|Q_i| = O(\frac{n\log n}{h_{i-1}})$
    using Lemma~\ref{lemma:Qj}) rounds.
    Step~\ref{algDirected:send} that involves all nodes $u\in V$ to 
    send distance values, $\delta^{h_i}(u,q_i)$ to all of their
    neighbor nodes take 
    $O(|Q_i|) = O(\frac{n\log n}{h_{i-1}})$ rounds.
    Steps~\ref{algDirected:local}-\ref{algDirected:update}
    are local steps and hence do not involve any communication.
    Since there are in total $O(\frac{\log n}{\log \log n})$
    iterations of the for loop, the entire execution of 
    Algorithm~\ref{algDirected} takes 
    $\tilde{O}(n)$ rounds in total. 
\end{proof}

Lemmas~\ref{lemma:algDirected:correctness} and \ref{lemma:algDirected:bound} together lead to the following
theorem.

\begin{theorem}
    There is a deterministic distributed algorithm that computes
    ANSC on an $n$-node directed
    graph with arbitrary non-negative edge weights in 
    $\tilde{O}(n)$ rounds.
\end{theorem}

\begin{corollary}
    There is a deterministic distributed algorithm that computes
    MWC on an $n$-node directed
    graph with arbitrary non-negative edge weights in 
    $\tilde{O}(n)$ rounds.
\end{corollary}
\section{Undirected Shortest Cycles}    \label{sec:undirected}

In this section, we describe our deterministic algorithm for 
computing MWC in a given 
undirected graph $G = (V, E)$.
As noted earlier in Section~\ref{sec:intro:undirected}
and has
also been noted in an earlier work on MWC by 
Roditty and Williams~\cite{RW11}, the
problem of computing shortest cycles in undirected graphs is 
much more difficult as compared to the directed graphs.
The reason behind this is that the reduction from MWC to APSP no
longer works in undirected graphs: an edge $(v,u)$ might also be 
the shortest path from $u$ to $v$, and 
$\min_u \{w(v,u) + \delta (u,v) \}$ might be equal to 
$2\cdot w(v,u)$ and not the weight of the shortest cycle 
containing node $v$.

A common approach for computing the value of the minimum weight
cycle in undirected graphs, is by taking the minimum 
across all possible values $\{w(v,u) + \delta(u,w) + \delta(v,w)\}$
where $u \in \mathcal{N}(v)$ and $w \in V$.
In the CONGEST model, this will require an initial computation of
APSP, followed by $O(n)$ communication rounds for each node to
communicate its distance values, $\delta(v,w)$, where $w \in V$, to
all its neighbor nodes.
This approach takes $O(APSP + n)$ rounds in total.
Since the current deterministic bound for APSP is 
$\tilde{O}(n^{4/3})$~\cite{AR20}, this approach would require 
$\tilde{O}(n^{4/3})$ rounds in total.

Similar to our algorithm for computing shortest cycles in
directed graphs described in Section~\ref{sec:directed}, our
algorithm for computing MWC in undirected graphs does not need 
an initial computation of APSP. 
Before further describing our approach for computing MWC in
undirected graphs, we first describe the notion of critical edge 
from Roditty and Williams~\cite{RW11} 
which is critical to obtaining our MWC algorithm.
Specifically given a minimum weight cycle $\mathcal{C}$ in graph $G$, the 
critical edge property states that for every node 
$v \in \mathcal{C}$, 
there exists an edge $(u,w)$ on $\mathcal{C}$ such that the 
disjoint paths from $v$ to $u$ and from $v$ to $w$ are 
both shortest paths in graph $G$.
We now state a simplified version of 
this property from~\cite{RW11} that we use in our algorithms below.

\begin{lemma}[Critical Edge~\cite{RW11}]    \label{lemma:critical-edge}
    Let $G = (V,E)$ be an undirected graph.
    Let $\mathcal{C} = \{v_1, v_2, \ldots, v_{l}\}$
    be a minimum weight cycle in $G$ of weight $wt(\mathcal{C})$ and 
    let $s \in \mathcal{C}$.
    There exists an edge $(v_i, v_{i+1})$ on $\mathcal{C}$ such that
    $\lceil \frac{wt(\mathcal{C})}{2} \rceil$ $-$ $w(v_i, v_{i+1}) \leq  \delta_{\mathcal{C}} (s, v_i) \leq \lfloor \frac{wt(\mathcal{C})}{2} \rfloor$
    and $\lceil \frac{wt(\mathcal{C})}{2} \rceil$ $-$ 
    $w(v_i, v_{i+1}) \leq \delta_{\mathcal{C}} (v_{i+1}, s) \leq \lfloor \frac{wt(\mathcal{C})}{2} \rfloor$.
    Furthermore, $\delta(s,v_i) = \delta_{\mathcal{C}} (s, v_i)$
    and $\delta(v_{i+1}, s) = \delta_{\mathcal{C}} (v_{i+1}, s)$.
\end{lemma}

Similar to the Lemma~\ref{property-cycle-directed} for 
shortest cycles in directed graphs,
using the above-described notion of critical edge
for undirected cycles we show that
there exists a minimum weight cycle $\mathcal{C}$ in an undirected 
graph $G$ (if $G$ has at least one cycle) 
such that $\mathcal{C}$ can be 
decomposed into two short-hop shortest paths from a node in 
a blocker set in the sequence $\{Q_i\}$, followed by a critical
edge.
Specifically, there exists a critical edge $(u,v) \in \mathcal{C}$
and a $j$ such 
that for some blocker node $q_j \in Q_j$, the cycle $\mathcal{C}$
can be decomposed into three components:
(1) $h_j$-hop shortest path from $q_j$ to $v$, 
(2) $h_j$-hop shortest path from $q_j$ to $u$,
(3) and the critical edge $(u,v)$.
We now establish this observation in the following lemma.

\begin{lemma}   \label{lemma:undirected}
    Let $G = (V,E)$ be an undirected graph 
    that contain at least one cycle and
    let $\{Q_i\}$ be the blocker set sequence for $G$.
    Then there exists a minimum weight cycle $\mathcal{C}$ in $G$
    and a blocker node $q_j \in Q_j$ for some $j$,
    such that 
    $wt(\mathcal{C})$ is equal to the sum of these three components:
    (1) weight of a critical edge $(u,v) \in \mathcal{C}$, $w(u,v)$.
    (2) weight of the $h_j$-hop shortest path from $q_j$ to $u$,
    $\delta^{h_j} (q_j, u)$, and
    (3) weight of the $h_j$-hop shortest path from $q_j$ to $v$,
    $\delta^{h_j} (q_j, v)$.
\end{lemma}

\begin{proof}
    On the contrary, assume that the above property does not hold 
    for any minimum weight cycles in some graph $G$.
    Let $j'$ be the maximum value of $j$ such that a blocker node
    $q_j$ from the blocker set $Q_j$ lies on a minimum weight 
    cycle in $G$, i.e. no other minimum weight cycles in $G$
    contains a blocker node from any of the blocker sets $Q_j$,
     where $j > j'$.
    Let this minimum weight 
    cycle be $\mathcal{C}'$ and $q_{j'}$ be a blocker node
    from blocker set 
    $Q_{j'}$ such that $q_{j'} \in \mathcal{C}'$. 

    Let $(u,v) \in \mathcal{C}'$ be the critical edge corresponding 
    to source $q_{j'}$.
    Using Lemma~\ref{lemma:critical-edge}, 
    $\delta(q_{j'}, u) = \delta_{\mathcal{C}'}(q_{j'}, u)$ and 
    $\delta(q_{j'}, v) = \delta_{\mathcal{C}'}(q_{j'}, v)$.
    Since the above property does not hold for any minimum weight
    cycle in $G$, at least one of the following conditions does not
    hold:
    
    \begin{enumerate}
        \item[(i)] $\delta^{h_{j'}}(q_{j'}, u) = \delta (q_{j'}, u)$
        \item[(ii)] $\delta^{h_{j'}}(q_{j'}, v) = \delta (q_{j'}, v)$
    \end{enumerate}    

    Without loss of generality assume that 
    $\delta^{h_{j'}}(q_{j'}, u) \neq \delta (q_{j'}, u)$.
    It implies that any shortest path from $q_{j'}$ to $u$ has
    hop length greater than $h_{j'}$.
    Let $p_{q_{j'}, u}$ be a shortest path from $q_{j'}$ to $u$
    that has minimum hop length and let $a$ be the node that is 
    at hop length $h_{j'}$ from node $q_{j'}$ on path 
    $p_{q_{j'}, u}$.

    Consider the path $p_{q_{j'}, a}$ that is part of the 
    collection $\mathcal{T}_{j'}$ constructed for source set 
    $Q_{j'}$ in Step~\ref{algBlocker:CSSSP} of 
    Algorithm~\ref{algBlocker}.
    By the definition of blocker set (Definition~\ref{def:blocker}),
    there exists a node $q_{j'+1} \in Q_{j'+1}$ such that it lies
    on the path $p_{q_{j'}, a}$.

    Consider the cycle $\mathcal{C}^{''}$ obtained by replacing the 
    path from $q_{j'}$ to node $a$ in $\mathcal{C}'$ with the path
    $p_{q_{j'}, a}$ .
    Clearly, the cycle $\mathcal{C}^{''}$ is also a minimum
    weight cycle in $G$.
    Now since the node $q_{j'+1} \in Q_{j'+1}$ lies on 
    $\mathcal{C}^{''}$, it violates our assumption that $j'$ is the 
    maximum value of $j$ such that a node from the blocker set $Q_j$
    lies on a minimum weight cycle in $G$, thus resulting in a 
    contradiction.

    This establishes the lemma.
    
\end{proof}

As noted earlier in the section, for a given source node $s$
and an edge $(u,v)$ the value 
$\{w(u,v) + \delta(s,u) + \delta(s,v)\}$ may not necessarily 
correspond to the weight of a cycle involving edge $(u,v)$.
It may be possible that either the shortest path from $s$ to $u$
or from $s$ to $v$ in itself contains edge $(u,v)$.
To differentiate such scenarios, we define a set 
$\mathcal{A}_{s,v} \subseteq \mathcal{N}(v)$ such that it contains
all nodes $u \in \mathcal{N}(v)$ such that the 
computed shortest paths from
$s$ to $u$ and from $s$ to $v$ does not contain edge $(u,v)$.
In a similar vein, we define the notion of $\mathcal{A}^h_{s,v}$
when we deal with $h$-hop shortest paths. Formally,
we define the sets $\mathcal{A}_{s,v}$ and $\mathcal{A}^h_{s,v}$
below:
$$\mathcal{A}_{s,v} = \{u | u\in \mathcal{N}(v), parent(s,v) \neq u 
 \text{\hspace{5pt} \& \hspace{5pt}}   parent(s,u) \neq v \}$$
$$\mathcal{A}^h_{s,v} = \{u | u\in \mathcal{N}(v), parent^h(s,v) \neq u  \text{\hspace{5pt} \& \hspace{5pt}} parent^h(s,u) \neq v \}$$

Given a source node $s \in V$, a node $v \in V$ can construct 
the set $\mathcal{A}^h_{s,v}$ locally, if it knows the 
$h$-hop shortest path distance values, $\delta^{h}(s,u)$, and the 
corresponding parent values, $parent^{h}(s,u)$, for all its 
neighbor nodes $u \in \mathcal{N}(v)$.
This leads to the following lemma.

\begin{lemma}
Let $G = (V,E)$ be an undirected graph.
Let $s \in V$ be a source node and let $h > 0$.
If every node $v \in V$ knows the $h$-hop shortest path distance
values, $\delta^{h}(s,u)$, and the corresponding parent values,
$parent^{h}(s,u)$, for all neighbor nodes $u \in \mathcal{N}(v).$
Then node $v$ can compute the set $\mathcal{A}^h_{s,v}$ locally
without the need of any further communication.
\end{lemma}

We now define below 
some new notations in order to accurately capture the 
set of cycles in graph $G$
such that they have a particular edge, say edge $(v,u)$, as the 
critical edge.

\begin{definition}
Let $G = (V,E)$ be an undirected graph and let
$\{Q_i\}$ be the blocker set sequence for $G$.
For a node $v \in V$, we define the following collections of 
cycles: 
\begin{enumerate}
    \item $\mathcal{C}^{h}_{s,v,u}$: collection of cycles $C$
in $G$ such that edge $(v,u)$ is the critical edge in $C$ with
node $s$ as the source node and the paths from $s$ to $v$ and 
from $s$ to $u$ are the $h$-hop shortest paths.
    \item $\mathcal{C}^{h}_{s,v}$: union of all sets 
$\mathcal{C}^{h}_{s,v,u}$ such that $u \in \mathcal{A}^h_{s,v}$.
    \item $\mathcal{C}^{h}_{Q_i,v}$: union of all sets 
    $\mathcal{C}^{h}_{q_j,v}$ such that $q_j \in Q_j$, where $Q_j$ is a blocker set in $\{Q_i\}$.
    \item $\mathcal{C}_{v}$: union of all sets $\mathcal{C}^{h_i}_{Q_i,v}$ for all 
    $0 \leq i \leq \lceil \frac{\log n}{\log \log n} \rceil - 1$
\end{enumerate}
\end{definition}

Algorithm~\ref{algUndirected} describes our deterministic algorithm
for computing the weight of the global minimum weight cycle in a 
given undirected graph $G = (V,E)$.
Let $\delta_{MWC}(v)$ refers to the minimum weight across all the 
cycles in the collection $\mathcal{C}_{v}$.
In Step~\ref{algUndirected:set} we set the initial hop length $h_0$
to $\log^2 n$ and we also set the value $\delta_{MWC}(v)$ to 
$\infty$.
We then compute the deterministic blocker set sequence $\{Q_i\}$
for graph $G$ using the algorithm described in 
Section~\ref{sec:blocker} in Step~\ref{algUndirected:blocker}.
The for loop in 
Steps~\ref{algUndirected:startFor}-\ref{algUndirected:endFor}
runs for $\lceil \frac{\log n}{\log \log n} \rceil$ iterations in
total.

\begin{algorithm}
\caption{Computing Undirected Minimum Weight Cycle}
\begin{algorithmic}[1]
\Statex \textbf{Input:} Undirected Graph $G = (V, E)$;
\State Set $h_0 \leftarrow \log^2 n$; $\delta_{MWC}(v) \leftarrow \infty$ \label{algUndirected:set}
\State Compute deterministic blocker set sequence $\{Q_i\}$ for $G$ using Algorithm~\ref{algBlocker}. \label{algUndirected:blocker}
\For{$0 \leq i \leq \lceil \frac{\log n}{\log \log n} \rceil - 1$}  \label{algUndirected:startFor}
\State $\begin{aligned}[t]
    & \text{Compute $h_i$-hop SSSP for all source nodes $q_i \in Q_i$ using Bellman-Ford algorithm~\cite{B58}} \\
    & \text{(Algorithm~\ref{algBF}).}
    \end{aligned}$ \label{algUndirected:outHop}
\State $\begin{aligned}[t]
    & \text{Each node $u \in V$ sends to all its neighbors $v \in \mathcal{N}(u)$, the $h_i$-hop SSSP distance values,} \\
    & \text{$\delta^{h_i}(q_i, u)$, along with the corresponding parent values, $parent^{h_i} (q_i, u)$ for all source nodes} \\
    & \text{$q_i \in Q_i$.}
    \end{aligned}$ \label{algUndirected:send}
\State \textbf{Local Step at each node $\bm{v \in V}$:} \label{algUndirected:local}
\State \hspace{10pt} \textbf{For each blocker node $q_i \in Q_i$:}
\State \hspace{20pt} Construct set $\mathcal{A}^{h_i}_{q_i, v}$ locally.    \label{algUndirected:construct}
\State \hspace{20pt} $\delta_{MWC}(v) \leftarrow \min \{\delta_{MWC} (v), \text{\hspace{5pt}}\min_{u \in \mathcal{A}^{h_i}_{q_i, v}} (w(v,u) + \delta^{h_i}(q_i, u) + \delta^{h_i}(q_i,v)) \} $ \label{algUndirected:deltav}
\State $h_i \leftarrow h_{i-1}\cdot \log n$ \label{algUndirected:update}
\EndFor \label{algUndirected:endFor}
\State Every node $v \in V$ broadcast their local value of $\delta_{MWC}(v)$ to the entire network. \label{algUndirected:broadcast}
\State Locally compute global minimum weight cycle by taking the minimum of all received values in Step~\ref{algUndirected:broadcast}.  \label{algUndirected:minimum}
\end{algorithmic}  \label{algUndirected}
\end{algorithm} 

Each iteration $i$ of the for loop in 
Steps~\ref{algUndirected:startFor}-\ref{algUndirected:endFor}
proceeds in the following manner:
Step~\ref{algUndirected:outHop} computes the $h_i$-hop shortest 
path distance values and the corresponding parent values from all 
source nodes $q_i$ in the blocker
set $Q_i$ using Bellman-Ford algorithm~\cite{B58}(also described
in Algorithm~\ref{algBF}).
Then in Step~\ref{algUndirected:send} every node $u \in V$ sends 
all the computed $h_i$-hop shortest path distance values, 
$\delta^{h_i}(q_i, u)$, and their corresponding parent values, 
$parent^{h_i} (q_i, u)$ for all blocker nodes $q_i \in Q_i$
to all its neighbor nodes $u \in \mathcal{N}(v)$.

Using the received shortest path distance values, 
$\delta^{h_i}(q_i, u)$, and the corresponding parent values, 
$parent^{h_i} (q_i, u)$, for all blocker nodes $q_i \in Q_i$,
in Step~\ref{algUndirected:construct}
every node $v \in V$ then constructs the set 
$\mathcal{A}^{h_i}_{q_i, v}$ locally for each blocker node 
$q_i$ separately.
For each blocker node $q_i \in Q_i$, node $v \in V$ then computes
the weight of candidate shortest cycles with edge $(v,u)$ as the 
critical edge for source node $q_i$, where 
$u \in \mathcal{A}^{h_i}_{q_i, v}$.
These weights can be computed by summing up the following three
components:
(1) weight of the critical edge $(v,u)$,
(2) followed by the $h_i$-hop shortest path distance value, 
$\delta^{h_i}(q_i, v)$, from $q_i$ to $v$, and
(3) the $h_i$-hop shortest path distance value, 
$\delta^{h_i}(q_i, u)$, from $q_i$ to $u$.
It then updates the value of $\delta_{MWC}(v)$, if the weight of
one of these candidate cycles is smaller than the current value of
$\delta_{MWC}(v)$ (Step~\ref{algUndirected:deltav}).

Finally, after all nodes in $V$ have computed their local values 
for $\delta_{MWC}(v)$, each node then broadcast its locally computed
value for $\delta_{MWC}(v)$ to the entire network in 
Step~\ref{algUndirected:broadcast} in $O(n)$ rounds.
Using all the $\delta_{MWC}(v)$ values received in 
Step~\ref{algUndirected:broadcast}, every node in the network can 
now compute the weight of the minimum weight cycle by taking the 
minimum of all the received $\delta_{MWC}(v)$ values 
(Step~\ref{algUndirected:minimum}).

We now show that Algorithm~\ref{algUndirected} correctly computes
the weight of the minimum weight cycle in a given undirected 
graph $G = (V,E)$, if a cycle exists in $G$.

\begin{lemma}   \label{lemma:algUndirected:correctness}
    Let $G = (V, E)$ be an undirected graph.
    Then Algorithm~\ref{algUndirected} correctly computes the 
    weight of the minimum weight cycle in $G$, if there exists a
    cycle in $G$.
\end{lemma}

\begin{proof}
    We consider two cases:
    (i) when there is no cycle in $G$
    (ii) there exists at least one cycle in $G$.

    \textit{(i) there is no cycle in graph $G$:}
    In this case, we need to show that the value for 
    $\delta_{MWC}(v)$ remains at $\infty$ throughout the execution
    of the algorithm, for all nodes $v \in V$.
    On the contrary assume that $\delta_{MWC}(v)$ was instead set to
    some finite values for some nodes.
    Let $v$ be such a node.

    The only step where the $\delta_{MWC}(v)$ value is  updated is
    Step~\ref{algUndirected:deltav}, where it is updated to value
    $w(v,u) + \delta^{h_i}(q_i, u) + \delta^{h_i}(q_i,v)$, for
    some blocker node $q_i \in Q_i$ and 
    $u \in \mathcal{A}^{h_i}_{q_i, v}$.
    It further implies that both distance values, 
    $\delta^{h_i}(q_i, u)$ and $\delta^{h_i}(q_i,v)$, are finite.
    Let $p_{q_i, u}$ be the path that corresponds to the 
    distance value $\delta^{h_i}(q_i, u)$ and let 
    path $p_{q_i, v}$ corresponds to the distance value 
    $\delta^{h_i}(q_i,v)$.
    Since $u \in \mathcal{A}^{h_i}_{q_i, v}$, edge $(u,v)$ is not
    present on both paths $p_{q_i, u}$ and $p_{q_i, v}$, and hence
    we can generate a cycle in $G$ by combining the edge 
    $(v,u)$, followed by paths 
    $p_{q_i, u}$ and $p_{q_i, v}$, thus
    resulting in a contradiction.
    \newline

    \textit{(ii) there exists at least one cycle in $G$:}
    In this case we need to show that the value for $\delta_{MWC}$
    computed by the algorithm equals the weight of 
    a minimum weight cycle in $G$.
    By Lemma~\ref{lemma:undirected}, there exists a minimum weight
    cycle $C$ in $G$ such that for some blocker node $q_j \in Q_j$
    and the critical edge $(u,v) \in C$ for source $q_j$, the 
    following holds:
    $$wt(C) = w(u,v) + \delta^{h_j}(q_j, u) + \delta^{h_j}(q_j,v)$$

    Now since edge $(u,v)$ is not part of the $h_j$-hop shortest
    paths from $q_j$ to nodes $u$ and $v$, 
    $u \in \mathcal{A}^{h_j}_{q_j, v}$ and
    $v \in \mathcal{A}^{h_j}_{q_j, u}$.
    Hence the above value for $wt(C)$ would have been captured in 
    Step~\ref{algUndirected:deltav} of the Algorithm at nodes $u$
    and $v$, thus $\delta_{MWC} \leq wt(C)$.

    Now if the value of $\delta_{MWC}$ is strictly less than $w(C)$,
    then similar to the previous case (Case(i)), we can argue that
    there exists another cycle composed of a critical edge $(u,v)$,
    followed by a path from a blocker node $q_{i'}$ to $v$ of
    hop-length at most $h_{i'}$, and another path from $q_{i'}$ to 
    $v$ of hop-length at most $h_{i'}$ and has total weight less
    than $wt(C)$.
    This results in a contradiction as $C$ is a minimum weight cycle
    in $G$.
    This completes the proof.
\end{proof}

In the following lemma, we show that Algorithm~\ref{algUndirected}
runs in total $\tilde{O}(n)$ rounds.

\begin{lemma}   \label{lemma:algUndirected:bound}
    Algorithm~\ref{algUndirected} computes $\delta_{MWC}$
    deterministically in $\tilde{O}(n)$ rounds.
\end{lemma}

\begin{proof}
    Step~\ref{algUndirected:set} is a local step.
    Using Lemma~\ref{lemma:algBlocker:bound}, the blocker set
    sequence $\{Q_i\}$ can be computed deterministically in 
    $\tilde{O}(n)$ rounds.
    Step~\ref{algUndirected:broadcast} takes $O(n)$ rounds using
    Lemma~\ref{lemma:broadcast} and Step~\ref{algUndirected:minimum}
    is a local step.
    We now show that each iteration of the for loop in 
    Steps~\ref{algUndirected:startFor}-\ref{algUndirected:endFor}
    takes $\tilde{O}(n)$ rounds.

    Using Bellman-Ford algorithm~\cite{B58}, 
    Step~\ref{algUndirected:outHop} takes $O(|Q_i|\cdot h_i) = $
    $O(\frac{n\log n}{h_{i-1}}\cdot h_i) = O(n\log^2 n)$
    (since $|Q_i| = O(\frac{n\log n}{h_{i-1}})$ using 
    Lemma~\ref{lemma:Qj}) rounds.
    Step~\ref{algUndirected:send} that involves all nodes 
    $u \in V$ to send distance values, $\delta^{h_i}(q_i, u)$,
    and the parent values, $parent^{h_i}(q_i, u)$, to all of their
    neighbor nodes takes $O(|Q_i|) = O(\frac{n\log n}{h_{i-1}})$
    rounds.
    Steps~\ref{algUndirected:local}-\ref{algUndirected:endFor}
    are local steps and hence do not involve any communication.
    Since there are in total $O(\frac{\log n}{\log\log n})$
    iterations of the for loop, the entire execution of 
    Algorithm~\ref{algUndirected} takes $\tilde{O}(n)$ rounds in 
    total.
\end{proof}

Lemmas~\ref{lemma:algUndirected:correctness} and
\ref{lemma:algUndirected:bound} together lead to the following
theorem.

\begin{theorem}
    There is a deterministic distributed algorithm that computes
    MWC on an $n$-node undirected graph with
    arbitrary non-negative edge weights in $\tilde{O}(n)$ rounds.
\end{theorem}

One could argue whether the approach discussed in this section can lead to an 
$\tilde{O}(n)$ rounds deterministic ANSC algorithm for undirected graphs.
In order to do so, one would need the exact shortest path distance values
from every node $v$ to the endpoints of the critical edge of its corresponding shortest
cycle (instead of the current Step~\ref{algUndirected:deltav} of 
Algorithm~\ref{algUndirected}).
This, however, will require an initial computation of APSP.
\section{Deterministic MSSP in Congest Model}    \label{sec:MSSP}

In this section, we describe how 
to compute multi-source shortest path distances using the deterministic
blocker set sequence algorithm described in Section~\ref{sec:blocker}. 
Our deterministic algorithm runs in $\tilde{O}(n)$ rounds when the 
source set $S$ has size at most $\sqrt{n}$.
This problem has been studied in other models of distributed 
computing, such as Congested Clique, in a recent work by 
Elkin and Neiman~\cite{E22}.
We now formally describe the problem statement below.

\begin{definition}[MSSP~\cite{E22}]
Let $G = (V, E)$ be a directed or undirected graph.
Let $S \subseteq V$ be a source set.
The MSSP problem deals with computing 
the shortest path distance values, $\delta(s, t)$ for each $s \in S$
at all nodes $t \in V$.
\end{definition}

Prior to this work, no near-linear rounds deterministic 
algorithm was known for this particular problem when 
$|S| = \Omega(n^{\epsilon})$ for any positive constant $\epsilon > 0$.
Algorithm~\ref{algMSSP} describes our deterministic algorithm for
computing the shortest path distance values from source nodes in set 
$S$ to all nodes $v \in V$.
In Step~\ref{algMSSP:blocker} we compute the deterministic blocker
set sequence $\{Q_i\}$ for graph $G$ using the algorithm described in
Section~\ref{sec:blocker}.
We then set the initial hop length $h_0$ to $\log^2 n$ and we define
the value of $h_i$ to $h_{i-1}\cdot \log n$ in Step~\ref{algMSSP:set}.

\begin{algorithm}
\caption{MSSP}
\begin{algorithmic}[1]
\Statex \textbf{Input:} \hspace{5pt} Directed or Undirected Graph $G = (V, E)$; 
\State Compute deterministic blocker set sequence $\{Q_i\}$ for $G$ using Algorithm~\ref{algBlocker}.   \label{algMSSP:blocker}
\State $h_0 \leftarrow \log^2 n$; \hspace{5pt} $h_i \leftarrow h_{i-1}\cdot \log n$.    \label{algMSSP:set}
\State Compute $\sqrt{n}$-hop out-SSSP for all source nodes $s \in S$ using Bellman-Ford algorithm~\cite{B58} (Algorithm~\ref{algBF}).   \label{algMSSP:sqrtnHop}
\State \textbf{Local Step at each node $\bm{v \in V}$:} Set $\delta(s, v) \leftarrow \delta^{\sqrt{n}}(s,v)$ for all $s \in S$. \label{algMSSP:localSet}
\State Let $i'$ be the max value of $i$ such that $h_i \leq \frac{\sqrt{n}}{2}$.  \label{algMSSP:max}
\For{$i' \leq i \leq \lceil \frac{\log n}{\log \log n} \rceil - 1$}  \label{algMSSP:startFor}
\State $\begin{aligned}[t] 
    & \text{Compute $h_i$-hop in-SSSP for all source nodes $q \in Q_i$ using Bellman-Ford algorithm~\cite{B58}} \\
    & \text{(Algorithm~\ref{algBFIncoming}).}
    \end{aligned}$ \label{algMSSP:inSSSP}
\State $\begin{aligned}[t] 
    & \text{Compute $h_i$-hop out-SSSP for all source nodes $q \in Q_i$ using Bellman-Ford algorithm~\cite{B58}} \\
    & \text{(Algorithm~\ref{algBF}).}
    \end{aligned}$ \label{algMSSP:outSSSP}
\State  $\begin{aligned}[t] 
    & \text{Each node $s \in S$  broadcasts the $h_i$-hop in-SSSP distance values, $\delta^{h_i}(s, q_i)$, for all} \\
    & \text{blocker nodes $q_i \in Q_i$.}
    \end{aligned}$ \label{algMSSP:broadcast}
\State \textbf{Local Step at each node $v \in V$:}  \label{algMSSP:local}
\State \hspace{10pt} \textbf{For each source $s \in S$:}
\State \hspace{20pt} $\delta(s,v) \leftarrow \min \{ \delta(s,v), \min_{q_i \in Q_i} \{\delta^{h_i}(s,q_i) + \delta^{h_i}(q_i, v)\}\}$  \label{algMSSP:delta}
\State $h_i \leftarrow h_{i-1}\cdot \log n$ \label{algMSSP:update}
\EndFor \label{algMSSP:endFor}
\end{algorithmic}  \label{algMSSP}
\end{algorithm} 

In Step~\ref{algMSSP:sqrtnHop} we compute the $\sqrt{n}$-hop outgoing
shortest path distance values from all source nodes $s \in S$ using
Bellman-Ford algorithm~\cite{B58} (Algorithm~\ref{algBF}).
Then in Step~\ref{algMSSP:localSet} every node $v \in V$ sets its
current estimate for source nodes $s \in S$, $\delta(s, v)$, to
$\delta^{\sqrt{n}}(s,v)$.

Since we have already computed shortest path distance values for all 
shortest paths from source set $S$ of hop-length at most $\sqrt{n}$,
we need to focus on computing shortest path distance values for paths
that have hop length greater than $\sqrt{n}$.
In order to do so, we compute the maximum value of $i$ such that 
the value of $h_i$ is at most $\frac{\sqrt{n}}{2}$, and we denote this value 
of $i$ as $i'$ (Step~\ref{algMSSP:max}.
The for loop in Steps~\ref{algMSSP:startFor}-\ref{algMSSP:endFor}
runs for $\lceil \frac{\log n}{\log \log n} \rceil - i'$
iterations in total.

Each iteration $i$ of the for loop in 
Steps~\ref{algMSSP:startFor}-\ref{algMSSP:endFor} proceeds in the
following manner: 
Step~\ref{algMSSP:inSSSP} computes the $h_i$-hop incoming shortest
path distance values, $\delta^{h_i}(v, q_i)$, from all blocker nodes
$q_i$ in the blocker set $Q_i$ using Bellman-Ford algorithm~\cite{B58} 
(also described in Algorithm~\ref{algBFIncoming}).
Similarly in Step~\ref{algMSSP:outSSSP} we compute the $h_i$-hop
outgoing shortest path distance values, $\delta^{h_i}(q_i, v)$, from
all blocker nodes $q_i$ in the blocker set $Q_i$ using Bellman-Ford \
algorithm~\cite{B58} (Algorithm~\ref{algBF}).

In Step~\ref{algMSSP:broadcast} every node $s \in S$ broadcasts the
$h_i$-hop incoming shortest path distance values, 
$\delta^{h_i}(s, q_i)$, to the entire network.
Using the received shortest path distance values, 
$\delta^{h_i}(s, q_i)$, for all source nodes $s \in S$ and 
blocker nodes $q_i \in Q_i$, in 
Steps~\ref{algMSSP:local}-\ref{algMSSP:delta} every node $v \in V$
then computes the candidate shortest path distance values from nodes 
in $S$ by summing up the distance values, $\delta^{h_i}(s, q_i)$, and
$\delta^{h_i}(q_i, v)$ for every pair $(s, q_i)$.
It then updates the current value for $\delta(s, v)$ for all 
source nodes $s$ using the equation in Step~\ref{algMSSP:delta}.

We now show that, for each source $s \in S$ and $v \in V$, 
Algorithm~\ref{algMSSP} correctly computes the shortest path distance
values, $\delta(s,v)$.
We break this into two parts, depending on the hop-length of the 
shortest path between nodes $s$ and $v$.
In Lemma~\ref{lemma:algMSSP:case1} we handle the case when the 
shortest path from $s$ to $v$ has hop length at most $\sqrt{n}$
and in Lemma~\ref{lemma:algMSSP:case2} we handle the case when 
all all shortest paths from $s$ to $v$ have hop-length 
greater than $\sqrt{n}$.

\begin{lemma}   \label{lemma:algMSSP:case1}
Let $G = (V, E)$ be a directed or undirected graph and let 
$S \subseteq V$ be a source set of $V$ of size at most $\sqrt{n}$.
Then for each source $s \in S$ and node $v \in V$ such that there 
exists a shortest path from $s$ to $v$ of 
hop-length at most $\sqrt{n}$, 
Algorithm~\ref{algMSSP} correctly computes the weight of the
shortest paths from $s$ to $v$.
\end{lemma}

\begin{proof}
Fix a source node $s \in S$ and a node $v \in V$.
Here we need to show that the computed value for $\delta(s,v)$ is
equal to $\delta^{\sqrt{n}}(s,v)$.

Since we compute $\sqrt{n}$-hop shortest path distance values from 
all source nodes in $S$ using the Bellman-Ford algorithm in 
Step~\ref{algMSSP:sqrtnHop}, the value $\delta^{\sqrt{n}}(s,v)$
is set  for $\delta(s, v)$ in Step~\ref{algMSSP:localSet}.
Hence $\delta(s,v) \leq \delta^{\sqrt{n}}(s,v)$.

Now consider the case when $\delta(s,v) < \delta^{\sqrt{n}}(s,v)$,
then it implies that there exists an $i$ such that for some
$q_i \in Q_i$, the following holds:
$$\delta^{h_i}(s, q_i) + \delta^{h_i}(q_i, v) < \delta^{\sqrt{n}}(s,v)$$
which is a contradiction, as this implies that there exists a path
from  $s$ to $v$ through node $q_i$ that has weight less than
$\delta^{\sqrt{n}}(s,v)$.
This establishes the lemma.
\end{proof}

\begin{lemma}   \label{lemma:algMSSP:case2}
Let $G = (V, E)$ be a directed or undirected graph and let 
$S \subseteq V$ be a source set of $V$ of size at most $\sqrt{n}$.
Then for each source $s \in S$ and node $v \in V$ such that all 
shortest paths from $s$ to $v$ have 
hop-length greater than $\sqrt{n}$, 
Algorithm~\ref{algMSSP} correctly computes the
shortest path distance values, $\delta(s,v)$, from nodes $s$ to $v$.
\end{lemma}

\begin{proof}
Fix a source node $s \in S$ and a node $v \in V$ such that all 
shortest paths from $s$ to $v$ have 
hop-length greater than $\sqrt{n}$.
By Lemma~\ref{lemma:property:Qj} there exists a $j$ such that 
$\delta(s,v) = \delta^{h_j}(s, q_j) + \delta^{h_j}(q_j, v)$.
Clearly $h_j > \frac{\sqrt{n}}{2}$, otherwise the hop-length of the 
path will be less than $\sqrt{n}$, and hence $j \geq i'$.
Here we need to show that the computed value for $\delta(s,v)$ is
equal to $\delta^{h_j}(s, q_j) + \delta^{h_j}(q_j, v)$.

Since we compute $h_j$-hop outgoing and incoming
shortest path distance values from 
all blocker nodes in $Q_j$ using the Bellman-Ford algorithm in 
Steps~\ref{algMSSP:inSSSP}-\ref{algMSSP:outSSSP}, 
the value $\delta^{h_j}(s, q_j) + \delta^{h_j}(q_j, v)$
is set  for $\delta(s, v)$ in Step~\ref{algMSSP:update}.
Hence $\delta(s,v) \leq \delta^{h_j}(s, q_j) + \delta^{h_j}(q_j, v)$.

Now consider the case when $\delta(s,v) < \delta^{h_j}(s, q_j) + \delta^{h_j}(q_j, v)$,
then it implies following two cases:
\newline

\textit{(i) that there exists an $i$ such that for some
$q_i \in Q_i$, the following holds:}
$$\delta^{h_i}(s, q_i) + \delta^{h_i}(q_i, v) < \delta^{h_j}(s, q_j) + \delta^{h_j}(q_j, v)$$

This implies that there is a shorter path from $s$ to $t$ through a 
blocker node $q_i \in Q_i$ for $i \neq j$, 
which is a contradiction as the path
from  $s$ to $v$ through node $q_j$ is the shortest path from $s$ to 
$v$.
\newline

\textit{(ii) $\delta^{\sqrt{n}}(s,v) < \delta^{h_j}(s, q_j) + \delta^{h_j}(q_j, v)$:}
This implies that there is a shorter path from $s$ to $t$ of
hop-length at most $\sqrt{n}$.
But this cannot happen since all shortest paths from $s$ to $v$
have hop-length greater than $\sqrt{n}$.
\end{proof}

In the following lemma, we show that Algorithm~\ref{algMSSP}
runs in $\tilde{O}(n)$ rounds.

\begin{lemma}   \label{lemma:MSSP:bound}
    Algorithm~\ref{algMSSP} computes all the $\delta(s,v)$ values
    deterministically in $\tilde{O}(n)$ rounds.
\end{lemma}

\begin{proof}
    Step~\ref{algMSSP:blocker} takes $\tilde{O}(n)$ rounds using
    Lemma~\ref{lemma:algBlocker:bound}.
    Steps~\ref{algMSSP:set} and \ref{algMSSP:localSet} are 
    local steps.

    Using Bellman-Ford algorithm~\cite{B58}, 
    Step~\ref{algMSSP:sqrtnHop} takes $O(\sqrt{n}\cdot \sqrt{n})$
    $= O(n)$ rounds.
    We now show that each iteration of the for loop in 
    Steps~\ref{algMSSP:startFor}-\ref{algMSSP:endFor} takes 
    $\tilde{O}(n)$ rounds.

    Steps~\ref{algMSSP:inSSSP} and \ref{algMSSP:outSSSP} takes
    $O(|Q_i|\cdot h_i) = O(\frac{n\log n}{h_{i-1}} \cdot h_i)$
    $= O(n\log^2 n)$ (since $|Q_i| = O(\frac{n\log n}{h_{i-1}})$
    using Lemma~\ref{lemma:Qj}) rounds.
    Step~\ref{algMSSP:broadcast} that involves all source nodes
    $s \in S$ to send distance values, $\delta^{h_i}(s,q_i)$, to
    the entire network takes $O(|S|\cdot |Q_i|)$ 
    $ = O(\sqrt{n}\cdot \frac{n\log n}{h_{i-1}})$
    $= O(\sqrt{n} \cdot \frac{n\log n}{\sqrt{n}/2}) $
    $= O(n\log n)$ rounds.
    Steps~\ref{algMSSP:local}-\ref{algMSSP:endFor} are local steps
    and hence do not involve any communication.
    Since there are in total $O(\frac{\log n}{\log \log n})$
    iterations of the for loop, the entire execution of the 
    algorithm takes $\tilde{O}(n)$ rounds in total.
\end{proof}

Lemmas~\ref{lemma:algMSSP:case1}-\ref{lemma:MSSP:bound} 
together lead to the following theorem.

\begin{theorem}
There is a deterministic distributed algorithm that computes 
MSSP on an $n$-node directed or
undirected graph with arbitrary non-negative edge weights in 
$\tilde{O}(n)$ rounds, given that the size of source set is 
at most $\sqrt{n}$.
\end{theorem}

\bibliographystyle{abbrv}
\bibliography{main}

\begin{thebibliography}{10}

\bibitem{A18}
U.~Agarwal and V.~Ramachandran.
\newblock Fine-grained complexity for sparse graphs.
\newblock In {\em Proceedings of the 50th Annual ACM SIGACT Symposium on Theory of Computing}, pages 239--252, 2018.

\bibitem{A19}
U.~Agarwal and V.~Ramachandran.
\newblock Distributed weighted all pairs shortest paths through pipelining.
\newblock In {\em 2019 IEEE International Parallel and Distributed Processing Symposium (IPDPS)}, pages 23--32. IEEE, 2019.

\bibitem{AR20}
U.~Agarwal and V.~Ramachandran.
\newblock Faster deterministic all pairs shortest paths in congest model.
\newblock In {\em Proceedings of the 32nd ACM Symposium on Parallelism in Algorithms and Architectures}, pages 11--21, 2020.

\bibitem{ARKP18}
U.~Agarwal, V.~Ramachandran, V.~King, and M.~Pontecorvi.
\newblock A deterministic distributed algorithm for exact weighted all-pairs shortest paths in {\~o} (n 3/2) rounds.
\newblock In {\em Proceedings of the 2018 ACM Symposium on Principles of Distributed Computing}, pages 199--205, 2018.

\bibitem{B58}
R.~Bellman.
\newblock On a routing problem.
\newblock {\em Quarterly of applied mathematics}, 16(1):87--90, 1958.

\bibitem{B94}
B.~Berger, J.~Rompel, and P.~W. Shor.
\newblock Efficient nc algorithms for set cover with applications to learning and geometry.
\newblock {\em Journal of Computer and System Sciences}, 49(3):454--477, 1994.

\bibitem{B21}
A.~Bernstein and D.~Nanongkai.
\newblock Distributed exact weighted all-pairs shortest paths in randomized near-linear time.
\newblock {\em SIAM Journal on Computing}, 52(2):STOC19--112, 2021.

\bibitem{C15}
M.~Cygan, H.~N. Gabow, and P.~Sankowski.
\newblock Algorithmic applications of baur-strassen’s theorem: Shortest cycles, diameter, and matchings.
\newblock {\em Journal of the ACM (JACM)}, 62(4):1--30, 2015.

\bibitem{D22}
M.~Dalirooyfard, C.~Jin, V.~V. Williams, and N.~Wein.
\newblock Approximation algorithms and hardness for $ n $-pairs shortest paths and all-nodes shortest cycles.
\newblock {\em arXiv preprint arXiv:2204.03076}, 2022.

\bibitem{D20}
M.~Dalirrooyfard and V.~V. Williams.
\newblock Conditionally optimal approximation algorithms for the girth of a directed graph.
\newblock In {\em 47th International Colloquium on Automata, Languages, and Programming (ICALP 2020)}, volume 168, page~35. Schloss Dagstuhl--Leibniz-Zentrum f $\{$$\backslash$" u$\}$ r Informatik, 2020.

\bibitem{E20}
M.~Elkin.
\newblock Distributed exact shortest paths in sublinear time.
\newblock {\em Journal of the ACM (JACM)}, 67(3):1--36, 2020.

\bibitem{E22}
M.~Elkin and O.~Neiman.
\newblock Centralized, parallel, and distributed multi-source shortest paths via hopsets and rectangular matrix multiplication.
\newblock In {\em 39th International Symposium on Theoretical Aspects of Computer Science}, 2022.

\bibitem{K22}
A.~Kadria, L.~Roditty, A.~Sidford, V.~V. Williams, and U.~Zwick.
\newblock Algorithmic trade-offs for girth approximation in undirected graphs.
\newblock In {\em Proceedings of the 2022 Annual ACM-SIAM Symposium on Discrete Algorithms (SODA)}, pages 1471--1492. SIAM, 2022.

\bibitem{K99}
V.~King.
\newblock Fully dynamic algorithms for maintaining all-pairs shortest paths and transitive closure in digraphs.
\newblock In {\em 40th Annual Symposium on Foundations of Computer Science (Cat. No. 99CB37039)}, pages 81--89. IEEE, 1999.

\bibitem{L72}
E.~L. Lawler.
\newblock Optimal cycles in graphs and the minimal cost-to-time ratio problem.
\newblock In {\em Periodic optimization}, pages 37--60. Springer, 1972.

\bibitem{L18}
A.~Lincoln, V.~V. Williams, and R.~Williams.
\newblock Tight hardness for shortest cycles and paths in sparse graphs.
\newblock In {\em Proceedings of the Twenty-Ninth Annual ACM-SIAM Symposium on Discrete Algorithms}, pages 1236--1252. SIAM, 2018.

\bibitem{L09}
A.~Lingas and E.-M. Lundell.
\newblock Efficient approximation algorithms for shortest cycles in undirected graphs.
\newblock {\em Information Processing Letters}, 109(10):493--498, 2009.

\bibitem{M22}
V.~Manoharan and V.~Ramachandran.
\newblock Brief announcement: Near optimal bounds for replacement paths and related problems in the congest model.
\newblock In {\em Proceedings of the 2022 ACM Symposium on Principles of Distributed Computing}, pages 369--371, 2022.

\bibitem{orlin2017nm}
J.~B. Orlin and A.~Sede{\~n}o-Noda.
\newblock An o (nm) time algorithm for finding the min length directed cycle in a graph.
\newblock In {\em Proceedings of the Twenty-Eighth Annual ACM-SIAM Symposium on Discrete Algorithms}, pages 1866--1879. SIAM, 2017.

\bibitem{P00}
D.~Peleg.
\newblock {\em Distributed computing: a locality-sensitive approach}.
\newblock SIAM, 2000.

\bibitem{R13}
L.~Roditty and R.~Tov.
\newblock Approximating the girth.
\newblock {\em ACM Transactions on Algorithms (TALG)}, 9(2):1--13, 2013.

\bibitem{RW11}
L.~Roditty and V.~V. Williams.
\newblock Minimum weight cycles and triangles: Equivalences and algorithms.
\newblock In {\em 2011 IEEE 52nd Annual Symposium on Foundations of Computer Science}, pages 180--189. IEEE, 2011.

\bibitem{R12}
L.~Roditty and V.~V. Williams.
\newblock Subquadratic time approximation algorithms for the girth.
\newblock In {\em Proceedings of the Twenty-Third Annual ACM-SIAM Symposium on Discrete Algorithms}, pages 833--845. SIAM, 2012.

\bibitem{S19}
P.~Sankowski and K.~W{\k{e}}grzycki.
\newblock Improved distance queries and cycle counting by frobenius normal form.
\newblock {\em Theory of Computing Systems}, 63(5):1049--1067, 2019.

\bibitem{Y11}
R.~Yuster.
\newblock A shortest cycle for each vertex of a graph.
\newblock {\em Information processing letters}, 111(21-22):1057--1061, 2011.

\end{thebibliography}

\end{document}